\documentclass[10pt,onecolumn,twoside]{IEEEtran}

\usepackage{amsmath,amssymb,euscript,yfonts,psfrag,latexsym,dsfont,graphicx,bbm,color,amstext,wasysym,pdfsync,thumbpdf,balance,flushend}
\usepackage[usenames,dvipsnames]{xcolor}

\definecolor{BrickRed}{rgb}{0.9,0,0}
\definecolor{RoyalBlue}{rgb}{0,0.5,0.7}
\definecolor{Gray}{rgb}{0.3,0.3,0.6}

\newtheorem{thm}{Theorem}

\newtheorem{lemma}[thm]{Lemma}
\newtheorem{prop}[thm]{Proposition}
\newtheorem{problem}{Problem}

\newtheorem{remark}{Remark}

\newcommand{\cK}{{\mathcal K}}
\newcommand{\mC}{{\mathbb C}}
\newcommand{\mJ}{{\mathbb J}}
\newcommand{\bJ}{{\mathbf J}}
\newcommand{\bPi}{{\mathbf \Pi}}
\newcommand{\la}{{z}}
\newcommand{\cH}{{\mathcal H}}
\newcommand{\bR}{{\mathbb R}}
\newcommand{\mD}{{\mathbb D}}
\newcommand{\bZ}{{\mathbb Z}}

\newcommand{\cP}{{\mathcal P}}
\newcommand{\cQ}{{\mathcal Q}}

\DeclareMathOperator{\tr}{tr}
\DeclareMathOperator{\range}{range}
\DeclareMathOperator{\argmin}{argmin}
\newcommand{\E}{\operatorname{E}}

\begin{document}

\title{Likelihood Analysis of Power Spectra\\
and Generalized Moment Problems\thanks{Supported in part by the
NSF under Grant ECCS-1509387,
the AFOSR under Grants FA9550-12-1-0319 and FA9550-15-1-0045,
the Vincentine Hermes-Luh Chair, 
the Swedish Research Council (VR) and the Foundation for Strategic Research (SSF).}}
\author{Tryphon T. Georgiou,~\IEEEmembership{Fellow,~IEEE} and Anders Lindquist,~\IEEEmembership{Life Fellow,~IEEE}
\thanks{T.T.\ Georgiou is with the Department of Electrical \& Computer Engineering,
University of Minnesota, Minneapolis, Minnesota, email: {\tt tryphon@umn.edu}; and A.\ Lindquist is with the Department of  Automation and the School of Mathematical Sciences, Shanghai Jiao Tong University, Shanghai, China, and with the Center for Industrial and Applied Mathematics and the Department of Mathematics,
KTH Royal Institute of Technology, Stockholm, Sweden, email: {\tt alq@kth.se}}
} 

\date{}
\maketitle
\begin{abstract}
We  develop an approach to
spectral estimation that has been advocated by Ferrante, Masiero and Pavon \cite{FMP} and, in the context of the scalar-valued covariance extension problem, by Enqvist and Karlsson \cite{EK}. The aim is to determine the power spectrum that is consistent with given moments and minimizes the relative entropy between the probability law of the underlying Gaussian stochastic process to that of a prior. The approach is analogous to the framework of earlier work by Byrnes, Georgiou and Lindquist and can also be viewed as a generalization of the classical work by Burg and Jaynes on the maximum entropy method. In the present paper we present a new fast algorithm in the general case (i.e., for general Gaussian priors)
and show that for priors with a specific structure the solution can be given in closed form.
\end{abstract}

\section{Introduction}

Consider a stationary, vector-valued, discrete-time, zero-mean, Gaussian stochastic process $\{y(t) \mid t\in\bZ\}$, where $y(t)\in\bR^m$, and $\bZ$, $\bR$ are the sets of integers and reals, respectively. We denote the corresponding probability law (on sample paths of the process) by $\cP$ \cite[Chapter 1]{pinsker} and
the power spectral density, which we assume exists, by $\Phi(e^{i\theta})$, $\theta\in[0,2\pi)$. Further, we assume that the stochastic process is nondeterministic in that the entropy rate is finite,
\[
\int_{-\pi}^\pi \log\det\Phi(e^{i\theta})d\theta < \infty.
\]
We study the basic problem to estimate $\Phi$ from sample-statistics of $\{y(t)\}$. Following \cite{FMP}, we view this problem in a large-deviations framework where a prior law $\cQ$ is available, and where this law corresponds to a power spectral density $\Psi$ with finite entropy rate. We postulate that available  sample-statistics of the process are not consistent with the prior law $\cQ$, and therefore we seek the law $\cP$ that {\em is} consistent with the available statistics and is the closest such law to the prior in the sense of large deviations, that is, $\cP$ is such that the Kullback-Leibler (KL) divergence between $\cP$ and $\cQ$ is minimal.
The KL divergence between the two laws is precisely the Itakura-Saito distance between the corresponding power spectra, which was considered in \cite{EK} for the special case of covariance extension for scalar-valued time series.

The theme of the approach, namely, to obtain power spectra that are consistent with given statistics and optimal with respect to a suitable criterion, is a standard recurring theme in works going back to Burg \cite{burg}. Statistics amount to (generalized) moment constraints and, in the past thirty or so years, a rich theory emerged that made contact with analytic function theory and the classical moment problem, see \cite{landau}--\cite{KLR} and the references therein. 
 A detailed and rigorous exposition of related topics and ideas in Signal Processing is given in \cite{stoicamoses}.

Initially, following Burg, early researchers focused on the entropy rate as such a suitable functional. The entropy rate relates to the variance of one-step-ahead linear prediction and the problem reduces to solving a linear set of equations--the normal equations \cite{makhoul}.  In the context of autoregressive modeling these are solved by the Levinson algorithm. To a large degree, research in the 80's was driven by problems in antenna arrays and, in particular, by interest in finding ``directions of arrival'' \cite{LangMcClellan1983,stoicali3,DOA1,DOA2,geo1} and the analysis of multidimensional processes \cite{lang1982multidimensional}. Two related themes emerged. First, power spectra were sought that consist of ``spectral lines'' (Pisarenko decomposition\footnote{\cite[Chapter 4.5]{stoicamoses}; see 
\cite{geo4,geo5,Stoica} for multivariable generalizations.}) and then, estimates of power over narrow frequency range was considered as a surrogate for power spectral density (Capon method\footnote{\cite[Chapter 5.4]{stoicamoses}; see 
\cite{geo2} for points of contact to the problem of moments.}).
Subsequent developments viewed spectral estimation as an inverse problem to achieve consistency with estimated statistics. Initial motivation was provided by a question of R.E. Kalman to identify spectra of low complexity  \cite{kalman}. %
Early results were obtained using topological and homotopy methods and the complete parametrization of solutions with generic minimal degree was formulated in steps in \cite{thesis,geo0} and \cite{BGuLM}. Subsequently, it was discovered that optimizers of
weighted entropy-like functionals (KL-divergence between power spectra as well as various types of distance to priors \cite{basseville}) had a particularly nice structure; they were rational and had small dimension \cite{BGuL}--\cite{KLR}, \cite{Karlson_Enqvist,continuous_input_to_state,THREE}.
In fact, it turned out that suitably specified weighted entropy functionals contained the precise degrees of freedom that were needed to efficiently parametrize and construct these generic minimal degree solutions \cite{BGuL,SIGEST,BEL,BGL,BGLM}. The mathematical underpinnings of this latter theory were largely based on optimization and duality. The present work is similar in spirit and  technique but differs substantially in the choice of functional and interpretation.

Following \cite{FMP,EK}, we consider the KL-divergence between {\em Gaussian probability laws} of stochastic processes or, equivalently, the Itakura-Saito distance between their power spectra. The interpretation as well as the structure of optimizers have subtle differences from earlier constructions.  For one thing, the use of the KL-divergence in this way has a very natural and appealing 
interpretation: the sought power spectra represent {\em the most likely statistical signature on the path space of a time series that is in agreement with the estimated sample statistics} (see Section \ref{sec:IIa}). The structure of solutions retains many of the attractive features of earlier works. In particular, it ensures  
reasonably good bounds on the dimensionality of modeling filters (see Remark \ref{remark3}). 

Below, we begin in Section \ref{sec:preliminaries} by discussing in some detail the motivation for choosing the particular functional to guide identifying suitable power spectra that reproduce sample statistics. We then explain how sample statistics impose moment constraints on sought power spectra. Finally we expand on a versatile structure for underlying dynamics (input-to-state filters) that allows an effective solution of the moment problems that arise in spectral analysis. In Section~\ref{sec:geometry} we present a geometric framework for input-to-state filters that provides basic tools for building a fast algorithm to solve the basic estimation problem.  Section \ref{sec:problem} gives the problem formulation and 
presents the main results. The fast algorithm is presented. For certain priors, we provide closed-form solutions, which correspond to autoregressive models in the case of trigonometric moment problems and all-pole priors.  The results are presented for multivariable time series and the moment problems for the corresponding matricial power spectra. 
Section~\ref{sec:example} provides a simple example and connections with earlier literature. Proofs of the main results are given in Sections \ref{sec:dual}-\ref{sec:closedform}. In particular, Section \ref{sec:fast} is devoted to deriving the fast algorithm and Section \ref{sec:closedform} to deriving the closed-form solution, respectively. In the concluding Section~\ref{sec:conclusions} we provide some final thoughts.

\section{Preliminaries}\label{sec:preliminaries}

\subsection{Likelihood framework}\label{sec:IIa}
The rationale for the framework adopted herein has been used to justify maximum likelihood methods \cite{JohnsonShore,Ellis,Georgii} and complements the original reasoning by E.T.~Jaynes \cite{Jaynes1,Jaynes2,Jaynes3}. It can be presented as follows. If sample paths of a time series are drawn out of the given prior $\cQ$, they have a small probability of giving rise to sample statistics that are not consistent with $\cQ$. 
If that were to happen, and thereby the sample paths represent a rare event, i.e., a departure from what is expected, one is motivated to seek out of the many possible sample-path distributions that are consistent with the observed statistics the one that is most likely. It is known that, asymptotically, the probability of rare events that suggest an (empirical) distribution $\cP$ depends exponentially on the KL divergence between the prior $\cQ$ and $\cP$.

The KL divergence between two laws $\cP$ and $\cQ$ is
\begin{equation}\label{eq:limit}
D(\cP\|\cQ)= \lim_{N\to \infty}
\frac{1}{2N+1} D(\cP|_{[-N,N]}\|
\cQ|_{[-N,N]}),
\end{equation}
where $\cP|_{[-N,N]}$ denotes the restriction of $\cP$ to the subset of random variables 
\[
\{y(-N),\ldots,y(-1),y(0),y(1),\ldots, y(N)\}
\]
and similarly for $\cQ|_{[-N,N]}$.
In turn, the KL divergence between the finite-dimensional probability densities $p(y(-N),\ldots, y(N))$ and $q(y(-N),\ldots, y(N))$,
corresponding to $\cP|_{[-N,N]}$ and $\cQ|_{[-N,N]}$, is
\[
\int_{\bR^{2N+1}} p \log(q/p)\;dy(-N)\cdots dy(N).
\]
Provided both laws represent purely nondeterministic processes, as is assumed herein, the limit in \eqref{eq:limit} exists. Using Szeg\"o-Wiener-Masani's formula (see e.g., \cite[Lemma 5.1]{Anderson etal.}, \cite[formula (E.12)]{SvS}, \cite[Theorem 11.3.5]{LPbook}), $D(\cP\|\cQ)$ can be 
expressed in terms of the corresponding power spectral densities as follows
\begin{align}
\label{eq:KLcriterion}
D(\cP\|\cQ)&=\frac{1}{4\pi}\int_{-\pi}^\pi \tr\left(\Phi\Psi^{-1}-\log\Phi\Psi^{-1} -I\right) d\theta\notag\\
&=: \mD(\Phi\|\Psi),
\end{align}
where $\tr(\cdot)$ denotes trace.
Since $\cP$ is completely specified by $\Phi$ we only need to determine $\Phi$, based of course on empirical statistics. Thus, we are interested in determining a power spectral density $\Phi$  that is consistent with given statistics and minimizes $\mD(\Phi\|\Psi)$ for a given power spectrum $\Psi$.  The precise formulation of the problem requires expressing statistics in terms of power spectra which is done next. The problem is stated precisely in Section \ref{sec:problem}.

\subsection{Filter banks and statistics}
Time-series represent samples of a stochastic process, and available statistics consist of sample covariances. We now explain the setting and nature of the covariance data.

In time-series analysis as well as in antenna array processing it is customary to assume that recorded data is scaled by a frequency-dependent vector/matrix-valued gain $G(e^{i\theta})$ where the frequency $\theta$ corresponds to time, space, angle, or even a vector-valued combination, see e.g., \cite[Chapter 6]{stoicamoses}, \cite{stoicali1,stoicali2}, \cite[Section II-B]{geo6}. For instance, a window of observations $\{y(k),\,y(k-1),\ldots,\,y(k-n)\}$ of a time series can be thought as the vectorial output of a ``tapped delay line'' represented by the vector-valued Fourier vector $\left[\begin{matrix}1& e^{-i\theta} & \ldots &e^{-in\theta}\end{matrix}\right]'$, i.e., the Fourier vector is the transfer function of the tapped delay line. Likewise, in the array processing literature, a model of an equispaced array of $n+1$ omnidirectional sensors registering signals that are emitted from afar is again the same Fourier vector \cite[Section 6]{stoicamoses}.
Such a vector-valued gain $G$, for general arrays, is often referred to as the {\em array manifold} and can be thought as a bank of filters that capture the relative dependence of the sensor outputs to signals from afar (see Fig.\  1).
Often, for a large equispaced array of sensors, a smaller output is selected that corresponds to $G$ being a linear combination of Fourier components (beamspace techniques)
\begin{equation}\label{eq:Fourier}
G(e^{i\theta})=M\left[\begin{matrix}1\\ e^{-i\theta}\\ \vdots \\e^{-i(n-1)\theta}\end{matrix}\right],
\end{equation}
for a suitable matrix $M$.
Other times, processing of time series or sensor-array data involves a suitably designed bank of filters $G_k(e^{i\theta})$, $k=1,2,\ldots,n$,
\begin{center}
\setlength{\unitlength}{.01in}
\parbox{304\unitlength}
{\begin{picture}(304,150)
\thicklines
\put(120,105){\framebox(50,30){$G_1$}}
\put(120,65){\framebox(50,30){$G_2$}}
\put(120,15){\framebox(50,30){$G_n$}}

\put(90,120){\vector(1,0){30}}
\put(90,80){\vector(1,0){30}}
\put(90,30){\vector(1,0){30}}

\put(90,30){\line(0,1){90}}
\put(80,75){\line(1,0){10}}

\put(170,120){\vector(1,0){30}}
\put(170,80){\vector(1,0){30}}
\put(170,30){\vector(1,0){30}}

\put(145,60){\circle{1}}
\put(145,55){\circle{1}}
\put(145,50){\circle{1}}

\put(185,130){\makebox(0,0){{\large $x_1$}}}
\put(185,90){\makebox(0,0){{\large $x_2$}}}
\put(185,40){\makebox(0,0){{\large $x_n$}}}

\put(70,75){\makebox(0,0){{\large $y$}}}
\end{picture}
}

Figure 1: Bank of filters\phantom{xxx}
\end{center}
in which case
\[
G(e^{i\theta})=\left[\begin{matrix}G_1(e^{i\theta}) & G_2(e^{i\theta}) & \ldots & G_n(e^{i\theta})\end{matrix}\right]'
\]
with $\{y(t)\}$ the common input and general dynamics, see e.g.\ \cite{PPV,BGL2}. The filters may also encapsulate attenuation from the coordinate $\theta$ of ``sources'' generating $\{y(t)\}$ to the respective outputs of sensor array (cf.\ \cite[Section II]{geo6}).
In all these cases, it is natural to estimate covariance of the vectorial time series
\[
x(t)=\left[\begin{matrix}x_1(t) & x_2(t) & \ldots & x_n(t)\end{matrix}\right]'.
\]
This is typically the form of available statistics that we consider henceforth.

We assume that $G$ is a square-integrable, stable $n\times m$ transfer function. Then, the $n$-dimensional output process $\{ x(t)\mid t\in\bZ\}$ 
assumes a representation as a stochastic integral
 \begin{displaymath}
x(t)=\int_{-\pi}^\pi e^{-it\theta}G(e^{i\theta})\underbrace{\,W(e^{i\theta})\,d\hat{w}(\theta)\,}_{d\hat y(\theta)},
\end{displaymath}
where $\hat{w}$ is a Wiener process such that $\E\{d\hat{w}d\hat{w}^*\}=Id\theta/2\pi$. Here, $I$ is the identity matrix, $\E\{\;\}$ the expectation operator, and $W$ is a (minimum-phase) spectral factor of $\Phi$, i.e., $W(e^{i\theta})W(e^{i\theta})^*=\Phi(e^{i\theta})$, and therefore $d\hat y$ is the stochastic Fourier transform of $y$; see  e.g.\ \cite[Chapter 3]{LPbook}. 
It follows that the covariance of the (zero-mean) vectorial output $x(t)$ is
\begin{equation}\label{eq:moment}
\Sigma:= \E\{x(t)x(t)\}=\int G\Phi G^* 
\end{equation}
where, for economy of notation, we have suppressed the limits of integration and the normalized Lebesgue measure $d\theta/2\pi$, i.e., $\int$ denotes $\int_{-\pi}^\pi \frac{d\theta}{2\pi}$.
The value $\Sigma$ represents a matricial moment constraint on $\Phi$.
The problem that we consider below is, given $G$ and $\Sigma$, to determine suitable $\Phi$
satisfying \eqref{eq:moment}. 

\subsection{Input-to-state filters}
A special case of a filter bank of great interest is when this represents an input-to-state (stable) linear system
\begin{equation}\label{eq:I2S}
x(t+1)=Ax(t)+By(t), \;t\in\bZ,
\end{equation}
where $A\in\bR^{n\times n}$ and $B\in\bR^{n\times m}$. In that case, the transfer function of the
filter bank is
\begin{equation}
\label{eq:Grational}
G(z)=z(zI-A)^{-1}B.
\end{equation}
Throughout we assume that all the eigenvalues of $A$ are located in the open unit disc. Then 
\begin{displaymath}
\begin{split}
G(z)&=(I-z^{-1}A)^{-1}B\\
&= B+ABz^{-1}+A^2Bz^{-2}+A^3Bz^{-3}+\dots 
\end{split}
\end{displaymath}
for all $z$ such that $|z|>1$. Throughout, to insure that the complete state space is being reached and to avoid trivialities we assume that $(A,B)$ is a reachable pair, i.e.,
\begin{equation}
\nonumber 
\mbox{\rm rank\,} [B,AB,\cdots,A^{n-1}B] =n,
\end{equation}
and that $B$ is full column rank.
The use of such filter banks is the basis of a tunable method of spectral analysis that was introduced in \cite{BGL2} and is referred to as THREE.

The input-to-state structure in \eqref{eq:Grational} encompasses
Fourier vectors where $G_k(z):=z^{-(k-1)}$, $k=1,2,\dots,n$, in which case 
\begin{equation}
\label{eq:ABtoeplitz}
A=\begin{bmatrix}0&0&\cdots&0&0\\1&0&\cdots&0&0\\
0&1&\cdots&0&0\\
\vdots&\vdots&\ddots&\vdots&\vdots\\
0&0&\cdots&1&0\end{bmatrix},
\quad
B=\begin{bmatrix}1\\0\\\vdots\\0\\0\end{bmatrix}
\end{equation}
and the state covariance is Toeplitz, i.e.,
 \begin{equation}\label{eq:toeplitz}
 \Sigma:=E\{x(t)x(t)'\}=\big[c_{k-\ell}\big]_{k,\ell=1}^n
 \end{equation}
where $c_k:=\E\{y(t+k)y(t)\}$. Identifying a power spectral density $\Phi$ that is consistent with $\Sigma$ and the process model is precisely the problem that underlies subspace identification \cite{LPbook} and coincides with
the classical ``covariance extension'' or ``trigonometric moment'' problem.

On the other hand, first-order filters $G_k(z):=\frac{z}{z-p_k}$, 
$k=1,2,\dots,n$, (with $p_k\neq p_\ell$ for $k\neq \ell\,$) lead to
\begin{equation}
\label{eq:ABn-p}
A=\begin{bmatrix}p_1& & &\\ &p_2&&\\ & &\ddots&  \\ & & & p_n\end{bmatrix},
\quad
B=\begin{bmatrix}1\\1\\\vdots\\1\end{bmatrix}\, 
\end{equation}
and a state covariance matrix $\Sigma$ that has the structure of a Pick matrix; see \cite{GL1}.  

Finally, it is also seen that \eqref{eq:Grational} is of the form \eqref{eq:Fourier} where $M=[B,\,AB,\,\ldots]$. This matrix is finite when $A$ is nilpotent corresponding to ``moving average'' dynamics. 

\section{Geometry of input-to-state filters}\label{sec:geometry}
The (rational) input-to-state structure of $G(z)$ in \eqref{eq:Grational} imposes structural algebraic constraints on the covariance of $x(t)$. In addition to positive definiteness,
$\Sigma$ is completely characterized by belonging to the range of the integral operator
\begin{equation}
\label{Gamma}
\Gamma\, :\, \Phi\mapsto \Sigma=\int G\Phi G^*.
\end{equation}
This is a linear operator that takes $m\times m$ integrable matrix-valued functions $\Phi$ on the unit circle to symmetric matrices $\Sigma$.

The range of $\Gamma$ admits an algebraic characterization. In fact, it is shown in \cite{Georgiou2002} that  a symmetric $n\times n$ matrix $\Sigma$ belongs to $\range(\Gamma)$ if and only if 
\begin{equation}
\label{eq:Sigma}
\Sigma- A\Sigma A' =BH+H'B' 
\end{equation}
for some $m\times n$ matrix $H$. Equivalently,
\begin{equation}
\label{eq:Sigma2}
{\rm rank}\,\left[\begin{matrix}\Sigma- A\Sigma A' & B\\B' & 0\end{matrix}\right] ={\rm rank}\,\left[\begin{matrix}0 & B\\B' & 0\end{matrix}\right],
\end{equation}
where $0$ denotes a zero-matrix of appropriate size, 
is necessary and sufficient for solvability of \eqref{eq:Sigma}.
Moreover, there is a coercive, continuous spectral density $\Phi$ satisfying the generalized moment condition \eqref{eq:moment} if and only if $\Sigma$ is positive definite\footnote{The case where $\Sigma$ is only nonnegative definite is discussed fully in \cite{geo5}. In that case the spectral content may correspond to a singular spectral measure.} and satisfies \eqref{eq:Sigma} or, the equivalent condition \eqref{eq:Sigma2}.

The adjoint operator $\Gamma^*$ maps symmetric matrices into $m\times m$ integrable Hermitian matrix-valued functions on the unit circle, namely
\[\Gamma^*\,:\, \Lambda \mapsto G^*\Lambda G.
\]
The inner product in these two spaces, symmetric matrices and integrable Hermitian matrix-valued functions on the unit circle, relate as
\begin{align*}
\langle \Lambda, \Sigma\rangle&:=\tr( \Lambda \Sigma)\\
&=\tr\int G^*\Lambda G \Phi\\
&=:\langle G^*\Lambda G, \Phi\rangle.
\end{align*}

We also consider the operator
\[
\Theta\,:\, H \mapsto \Delta=BH+H'B'
\]
which maps $\bR^{m\times n}$ to symmetric $n\times n$ matrices and its adjoint
\begin{equation}
\label{eq:Thetastar}
\Theta^*\,:\, \Delta \mapsto 2B'\Delta.
\end{equation}
We are interested in non-redundant representations of $\range(\Gamma)$ and $\range(\Gamma^*)$ by identifying the minimal degrees of freedom in suitable matrix representations. The first proposition deals with $\range(\Gamma)$.

\

\begin{prop}\em\label{prop:bijection1} The map
\begin{subequations}\label{eq:HY}
\begin{equation}
\label{eq:H1}
\Sigma\mapsto H=(B'B)^{-1}\left[B'(\Sigma- A\Sigma A')-YB'\right]
\end{equation}
where $Y$ is the symmetric solution of the Lyapunov equation
\begin{equation}\label{eq:YLyapunov1}
(B'B)Y+Y(B'B)=B'(\Sigma- A\Sigma A') B
\end{equation}
establishes a bijective correspondence between
$\Sigma\in\range(\Gamma)$ and $H\in\range(\Theta^*)$.
\end{subequations}
\end{prop}

\

\begin{proof}
Set
$\Delta:=\Sigma-A\Sigma A'$. Since we have $\Sigma\in\range(\Gamma)$, 
\begin{equation}\label{eq:Delta}
\Delta= BH+H'B'
\end{equation}
can be solved for $H\in\bR^{m\times n}$ and $\Delta=\Theta(H)$. We seek a particular solution of minimal Frobenius norm 
\[
\|H\|_{\rm F}:=\sqrt{\tr(HH')}.
\]
Then, this solution will be in $\range(\Theta^*)$, a fact that will be verified below.
The Lagrangian of the problem is
\[
\tr(HH')+2\tr(\Lambda BH)-\tr(\Lambda\Delta)
\]
where $\Lambda=\Lambda'$ is the symmetric matrix-valued Lagrange multiplier. It follows that the unique optimal solution is of the form
\begin{equation}
\label{H=B'Lambda}
H=B'\Lambda,
\end{equation}
and therefore $H\in\range(\Theta^*)$ in view \eqref{eq:Thetastar}.
Then, $HB=B'\Lambda B=:Y$ is symmetric. Further, it satisfies the Lyapunov equation
\begin{equation}\label{eq:YLyapunov}
(B'B)Y+Y(B'B)=B'\Delta B.
\end{equation}
as can be seen by pre-multiplying \eqref{eq:Delta} by $B'$ and post-multiplying by $B$. Since $B$ has full column rank by assumption, the eigenvalues of $B'B$ are positive and \eqref{eq:YLyapunov} has a unique solution $Y$.
By premultiplying  \eqref{eq:Delta} by $B'$ we can now solve for
\begin{equation}
\label{eq:H}
H=(B'B)^{-1}\left(B'\Delta-YB'\right).
\end{equation}
Finally, suppose that \eqref{eq:Sigma} has two solutions $H_1$ and $H_2$ in $\range(\Theta^*)$. Then 
\begin{displaymath}
H_1-H_2\in \ker\Theta =\left(\range(\Theta^*)\right)^\perp ,
\end{displaymath}
and hence $H_1=H_2$, proving uniqueness.
\end{proof}

\

The essence is that
 \eqref{eq:Sigma} has many solutions in general when $m>1$. In that case, $\Theta$ has a non-trivial null space, and Proposition \ref{prop:bijection1} provides the solution to \eqref{eq:Sigma} 
of minimal Frobenius norm.

The next proposition deals with $\range(\Gamma^*)$. 
Since the orthogonal complement of the range of $\Gamma$ is the null space of $\Gamma^*$,
elements in  $\range(\Gamma^*)$ can always be written
in the form $G^*\Lambda G$ where $\Lambda\in \range(\Gamma)$.

\

\begin{prop}\em\label{prop:bijection2} The map
\begin{subequations}
\begin{equation}
\label{eq:}
G^*\Lambda G \mapsto X=MB,
\end{equation}
where $M$ is the unique solution of the Lyapunov equation
\begin{equation}
\label{eq:M1}
M=A'MA + \Lambda 
\end{equation}
establishes a bijective correspondence between
$G^*\Lambda G\in\range(\Gamma^*)$ and $X\in\range(\Theta^*)$.
\end{subequations}
\end{prop}

\

\begin{proof}
We first note that the dimension of $\range(\Gamma^*)$, which coincides with the dimension of $\range(\Gamma)$, is equal to the dimension of $\range(\Theta^*)$  by Proposition \ref{prop:bijection1}. Since $M$ is symmetric, it also follows that $X'=B'M\in\range(\Theta^*)$. Thus, in order to establish that the correspondence $G^*\Lambda G\mapsto X'$ is a bijection, it suffices to prove that $X'=0$ only when $\Lambda=0$. To see this note that, since $AG(z)=z\big(G(z) -B\big)$, \eqref{eq:Sigma} yields
\begin{equation}
\label{eq:Lambda2X}
\begin{split}
G^*\Lambda G &= G^*MG -G^*A'MA G\\
                                 &= G^*MG - [G -B]^*M[G -B] \\
                                 &= G_0^*X + X'G_0
\end{split}
\end{equation}
with  $G_0$ given 
\begin{equation}
\label{eq:Gtilde}
\begin{split}
G_0(z)&:=G(z) -\tfrac{1}{2} B = \tfrac12 B +A(zI-A)^{-1}B\\
            &= \tfrac{1}{2} B+ABz^{-1}+A^2Bz^{-2}+A^3Bz^{-3}+\dots .
\end{split}
\end{equation}
But, since $\Lambda\in\range(\Gamma)$, $G^*\Lambda G=0$ only when $\Lambda=0$. Thus, $X=0$ implies that $\Lambda=0$ and this completes the proof.
\end{proof}

\section{Main results}\label{sec:problem}\label{mainresults}

We are now in a position to formulate the main problem that we consider. As noted earlier this problem was first formulated and studied in \cite{FMP}. 

\begin{problem}\label{problem1}
 Given an $m\times m$ matrix-valued power spectral density $\Psi$, and given the parameters $A,B$ of the input-to-state filter (filter bank) in \eqref{eq:I2S} and the covariance $\Sigma$ of the state process $x(t)$,
determine 
\[
\hat\Phi\in\argmin\{\mD(\Phi\|\Psi)\mid \mbox{ such that }\eqref{eq:moment}\mbox{ holds}\}.
\]
\end{problem}
\

Next we discuss the structure of solutions. The expressions we give provide and alternative to those in \cite{FMP} and require fewer variables in general. This non-redundant structure of solutions is analogous to the reduction in the number of variables enabling the fast algorithms for Kalman filtering in \cite{fast}.


As in our previous work on the moment problem, e.g., \cite{SIGEST,BGL,BGL2}, solving Problem 1 reduces to convex optimization. 
With $G_0$ given by \eqref{eq:Gtilde}, the optimization criterion is the 
strictly convex functional 
\begin{equation}
\label{J(X)}
\begin{split}
\mathbf{J}(X)&=\tr \left\{(HX+X'H')\phantom{\int} \right.\\
&\left.-\int\log \big(\Psi^{-1}+G_0^*X + X'G_0\big)\right\},
\end{split}
\end{equation}
defined on the open set $\mathcal{X}_+$ of matrices $X\in\bR^{n\times m}$ such that
$X'\in \range(\Theta^*)$, i.e., $B'X$ is symmetric, and
\begin{equation}
\label{eq:Q(X)}
Q(z):=\Psi(z)^{-1}+G_0(z)^*X + X'G_0(z) 
\end{equation}
is positive definite at each point $z=e^{i\theta}$ on the unit circle. 

\

\begin{thm}\label{thm1}\em
Let $\Sigma$ be a symmetric, positive definite $n\times n$ matrix in the range of $\Gamma$, and let $H$ be given by \eqref{eq:HY}. Suppose that the prior spectral density $\Psi$ is coercive on the unit circle and that the components of $\theta\mapsto\Psi(e^{i\theta})^{-1}$ are Lipschitz continuous.
Then Problem 1 has the unique solution
\begin{subequations}\label{eq:optimalPhiy}
\begin{equation}
\label{eq:Phihat}
\hat\Phi=\hat{Q}^{-1},
\end{equation}
where 
\begin{equation}
\label{eq:Qhat}
\hat{Q}= \Psi^{-1}+G_0^*\hat{X} + \hat{X}'G_0
\end{equation}
for some $\hat X\in \mathcal{X}_+$.
\end{subequations}
The matrix $\hat X$ is the unique minimizer of the functional $\mathbf{J}(X)$, and it is also the unique solution of the stationarity condition
\begin{equation}
\label{eq:Xhat}
\int \big(\Psi^{-1}+G_0^*X + X'G_0\big)^{-1}G_0^*=H .
\end{equation}
\end{thm}

\bigskip

The solution in the theorem can be obtained numerically by 
a Newton method to compute the minimizer of $\mathbf{J}$. To this end,
we compute the gradient
\begin{equation}
\label{eq:gradient}
\frac12\frac{\partial\mathbf{J}}{\partial X}= H- \int Q^{-1}G_0^*,
\end{equation}
where $Q$ is given by \eqref{eq:Q(X)}, and the Hessian
\begin{equation}
\label{eq:Hessian}
\frac12\mathcal{H}(X)= \int G_0 Q^{-2}G_0^* >0.
\end{equation}
The positivity of the Hessian indeed shows that the functional $\mathbf{J}$ is strictly convex.
A possible starting point is $X=0$.


In the next theorem we consider the special case where the prior $\Psi$ has the form $\Psi=(G^*\Lambda_0 G)^{-1}$. Then the solution to Problem 1 can be given in closed form.

\begin{thm}\label{thm2}\em
Let $\Sigma$ be a positive definite $n\times n$ matrix in the range of $\Gamma$, and suppose that the prior $\Psi$ is given by \begin{equation}
\label{eq:Phiz}
\Psi=(G^*\Lambda_0 G)^{-1}\, .
\end{equation}
Then Problem 1 has the unique solution
\begin{equation}
\label{eq:closedform}
\hat\Phi=\big(G^*\hat\Lambda G\big)^{-1},
\end{equation}
where
\begin{equation}
\nonumber 
\hat\Lambda:=\Sigma^{-1}B(B'\Sigma^{-1} B)^{-1}B'\Sigma^{-1}
\end{equation}
and does not depend on $\Lambda_0$.
\end{thm}

\

The proofs of Theorems \ref{thm1} and \ref{thm2} are given in Sections \ref{sec:fast} and \ref{sec:closedform}, respectively.

\

\begin{remark} It is interesting to point out that the solution $\hat\Phi$ to Problem 1 shares the same zeros as the prior $\Psi$. To see this note that at any value on the complex plane where $\Psi$ becomes singular, $\hat Q$ becomes infinity along suitable direction, and therefore $\hat\Phi$ becomes singular as well. This property is also present in solutions to moment problems that minimize alternative entropy functionals and has been explored in our earlier work.
It is quite instructive to consider Problem 1 in the scalar case ($m=1$). Then the optimal solution takes the form
\begin{equation}
\label{eq:scalarsolution}
\hat\Phi(e^{i\theta})=\frac{\Psi(e^{i\theta})}{1+2\Psi(e^{i\theta})\text{Re}\{G_0(e^{i\theta})^*\hat{X}\}}.
\end{equation}
Any zeros of the prior $\Psi$ will therefore  be zeros also of $\hat\Phi$. However, in the special case of rational $\Psi$, the dimension of modeling filters corresponding to $\hat\Phi$ is enlarged as compared to alternative formulations in our earlier works, e.g., \cite{BGL,BGL2}.
\end{remark}

\

\begin{remark}
Since the closed-form solution \eqref{eq:closedform} does not depend on $\Lambda_0$, we may in particular choose $\Lambda_0=I$. Then, in the important case when $G_k(z):=z^{-(k-1)}$, $k=1,2,\dots,n$, we have $\Psi=I$, leading to an autoregressive (maximum-entropy) model.  
\end{remark}

\

\begin{remark}\label{remark3} 
Going back to \cite{kalman} the original motivation
was to identify and characterize solutions to moment problems having low degree. It is instructive to consider the scalar trigonometric moment problem with data \eqref{eq:ABtoeplitz} and \eqref{eq:toeplitz}, that is, the problem to match the $n$ covariance samples $\{c_0,\,c_1,\ldots,c_{n-1}\}$ with a rational power spectrum $\Phi$,
in the sense that
\[
c_k=\int e^{ik\theta}\Phi \mbox{ for } k=0,\,1,\ldots,n-1
\]
holds, or, equivalently, \eqref{eq:moment} holds for the $n\times n$ covariance matrix $\Sigma$. There is a generic set of covariance samples (i.e., a set with an open interior) for which the minimal degree solution has spectral factors of degree $n-1$ \cite{thesis}. (For a more general result of this type; see \cite[Theorem 2.2]{BLpartial97}.) The family of all power spectra with the same dimensionality can be parametrized by a set of arbitrarily selected $n-1$ spectral zeros (i.e., zeros of the corresponding minumum-phase spectral factor) -- existence of power spectra corresponding to each such choice was shown in \cite{thesis,geo0} and uniqueness was shown in \cite{BGuLM}. Likewise, in the case of $m$-vector valued time series where an $n\times n$ covariance matrix $\Sigma$ is available, the family of generically minimal degree solutions has spectral factors of degree $n-m$, parametrized accordingly for a choice of spectral zeros \cite[Section IV and Corollary 2]{takyar}). On the other hand, a direct approach of constructing solutions based on the  THREE framework gives a family of solutions with spectral factors of degree $n$ \cite[Section IV-B]{geo6} (instead of the generic minimum $n-m$ in \cite{takyar}) likewise parametrized by a suitable choice of spectral zeros. The current framework allows constructing solutions \eqref{eq:closedform} with spectral factors of degree $n$ only when the zero-structure is trivial (i.e., identical to the eigenvalues of the matrix $A$), while in general the best bound one can provide from \eqref{eq:optimalPhiy} for the dimension of spectral factors is $n+\tfrac{1}{2}\times$(degree of $\Psi$); cf.\ \cite[Section IV]{FMP}.
\end{remark}

\section{A simple example}\label{sec:example}
Consider the case where $\Sigma$ is a Toeplitz matrix as given in \eqref{eq:toeplitz} with covariance lags $c_k:=\E\{y(t+k)y(t)\}$   of a scalar stationary process $y$. Then $G$ is given by \eqref{eq:Fourier} with $M=I$. Moreover, $A$ and $B$ are given by \eqref{eq:ABtoeplitz}, and hence $B'(\Sigma- A\Sigma A)=(c_0,c_1,\dots,c_{n-1})$ and $B'B=1$. Consequently it follows from \eqref{eq:YLyapunov} that $Y=\frac12 c_0$ and from \eqref{eq:H} that $H=(\frac12 c_0,c_1,\dots,c_{n-1})$. Then, setting $X':=(q_0, q_1,\dots, q_{n-1})$, we have
\begin{displaymath}
HX+X'H'=\langle c,q\rangle :=\sum_{k=-(n-1)}^{n-1} c_kq_k
\end{displaymath}
and
\begin{subequations}\label{eq:all}
\begin{equation}\label{presence}
Q(e^{i\theta})=\Psi(e^{i\theta})^{-1} +\sum_{k=-(n-1)}^{n-1}q_ke^{ik\theta}.
\end{equation}
Problem 1 then amounts to minimizing
\begin{equation}\label{eq:Jagain}
\mathbf{J}(q)=\langle c,q\rangle -\int\log Q
\end{equation}
over all $q:=(q_0,q_1,\dots, q_{n-1})$ such that $Q(e^{i\theta})>0$ for all $\theta$. In this notation the stationarity condition \eqref{eq:Xhat} becomes
\begin{equation}
\label{cmoment}
\int_{-\pi}^\pi e^{ik\theta}Q^{-1}\frac{d\theta}{2\pi} =c_k, \quad k=0,1,\dots,n-1.
\end{equation}
\end{subequations}

\

\begin{remark}
It is interesting to compare the functional $\mathbf{J}(q)$ and the form of solution above to those in the framework of, e.g.,
\cite{SIGEST,BGL,BGL2}. There, the corresponding functional is
\[
\mathbf{J}(q)=\langle c,q\rangle -\int\Psi\log Q
\]
instead of \eqref{eq:Jagain}, with
\[
Q(e^{i\theta})=\sum_{k=-(n-1)}^{n-1}q_ke^{ik\theta}
\]
and moment conditions
\[
\int_{-\pi}^\pi e^{ik\theta}\frac{\Psi}{Q}\frac{d\theta}{2\pi} =c_k, \quad k=0,1,\dots,n-1,
\]
instead of \eqref{eq:all}.
We see that the present framework is analogous to the maximum-entropy solution in these earlier works except for the absence of $\Psi(e^{i\theta})^{-1}$ in the corresponding expression for $Q$ which is traded off with the direct presence of $\Psi$ in functional and the stationarity conditions. The optimal solution is
\[
\hat \Phi(e^{i\theta})=\frac{\Psi(e^{i\theta})}{\sum q_k e^{ik\theta}}
\]
in this case instead of 
\[
\hat \Phi(e^{i\theta})=\frac{\Psi(e^{i\theta})}{1+ \Psi(e^{i\theta}) \sum q_k e^{ik\theta}}
\] 
  in our present framework.
\end{remark}

\

We proceed with our numerical example. To this end, we select
\begin{align}\Phi(z)&=|(z+1)(z^{-1}+1)|^3  \label{true}\\
&= z^{-3}+ 6z^{-2}+ 15z^{-1}+   20 +15z+ 6z^2+ z^3\nonumber
\end{align}
that corresponds to a moving average filter with transfer function $W(z)=1+3z^{-1}+3z^{-2}+z^{-3}$.
In Fig.~2 we first compare the true power  spectral density $\Phi$ in \eqref{true}, evaluated at $z=e^{i\theta}$ for $\theta\in[0,\,\pi]$, with a prior $\Psi=10 (1+0.9 \cos(\theta)(1+0.9\cos(\theta))^2)$ that is selected
to have a low pass characteristic. We seek to match $8$ moments, namely,
$c=(20,\,15,\,6,\,1,\,0,\,0,\,0,\,0)$.
Next, in Fig.\ 3, we compare $\Phi$ with the optimal solution $\hat \Phi$ to Problem 1 for the given $\Psi$. Finally, in Fig.\ 4, we compare $\Phi$ with the solution corresponding to the choice $\Psi=1$.
The power spectral density obtained in this way, using either the Newton algorithm based on Theorem \ref{thm1} or the closed-form expression in Theorem \ref{thm2}, is an {\em all-pole} power spectrum that agrees with the given moments. 

\begin{center}
       \hspace*{-11pt}\includegraphics[width=0.52\textwidth]{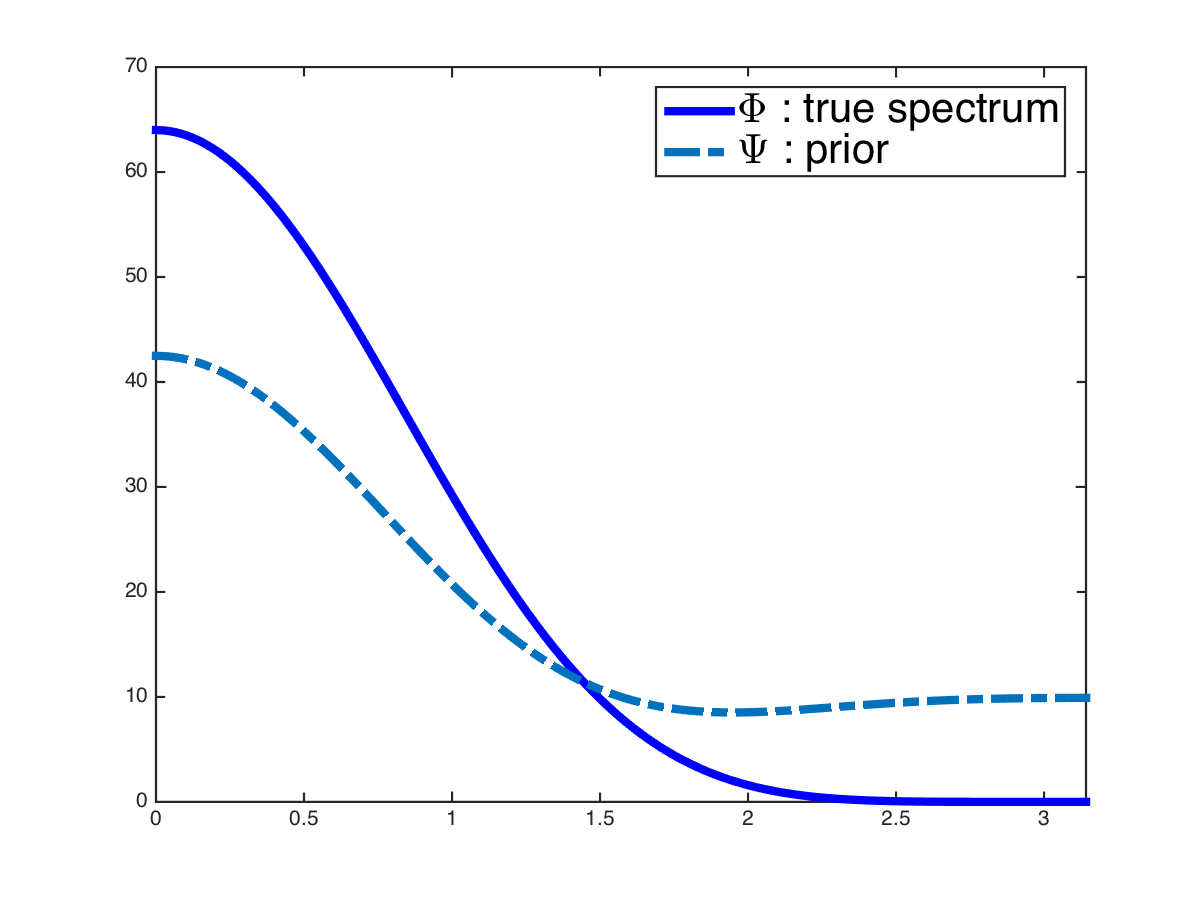}\\
      {\small Fig.\ 2. True spectrum (solid line) vs.\ prior (dashed)}
      \end{center}
      \begin{center}
       \hspace*{-11pt}\includegraphics[width=0.52\textwidth]{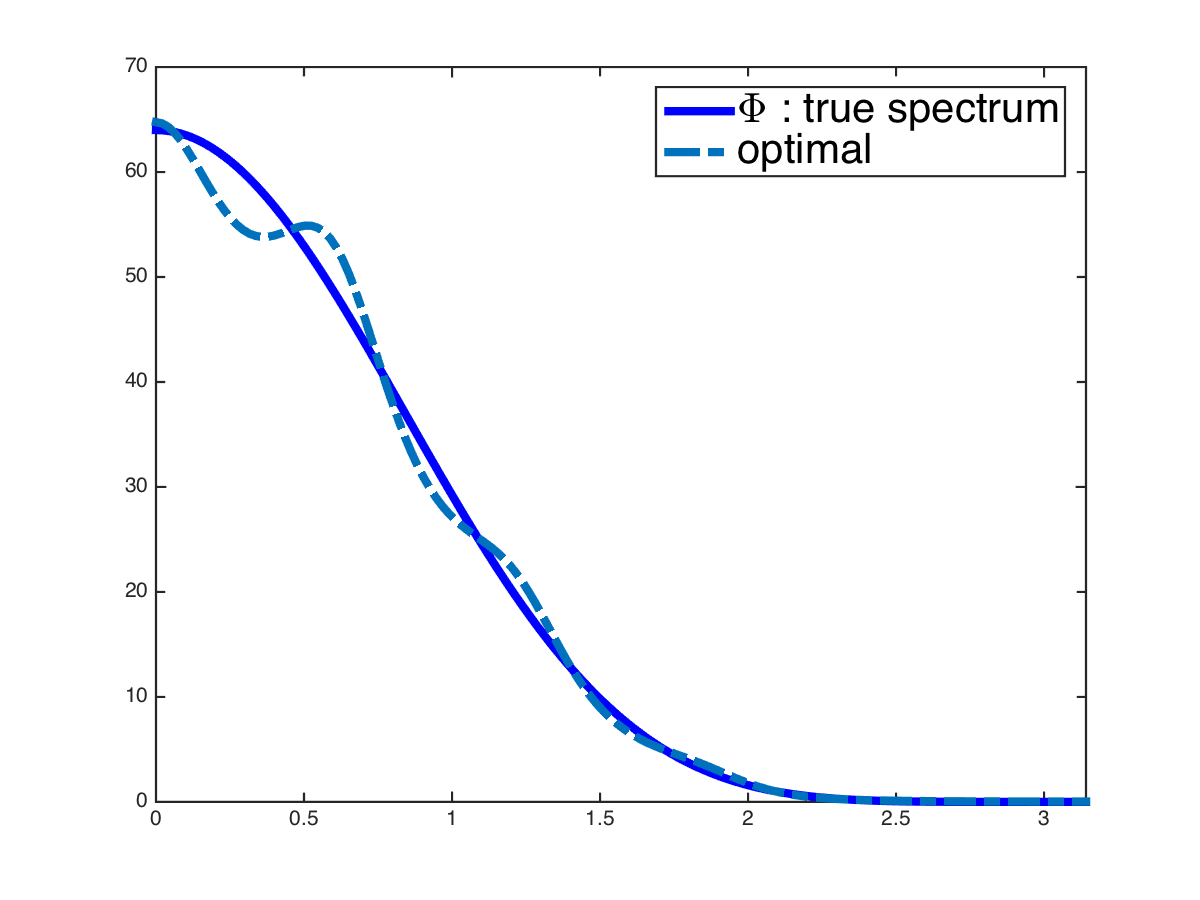}\\
  {\small Fig.\ 3. True spectrum (solid line) vs.\ optimal (dashed)}
  \end{center}
  
      \begin{center}
       \hspace*{-11pt}\includegraphics[width=0.52\textwidth]{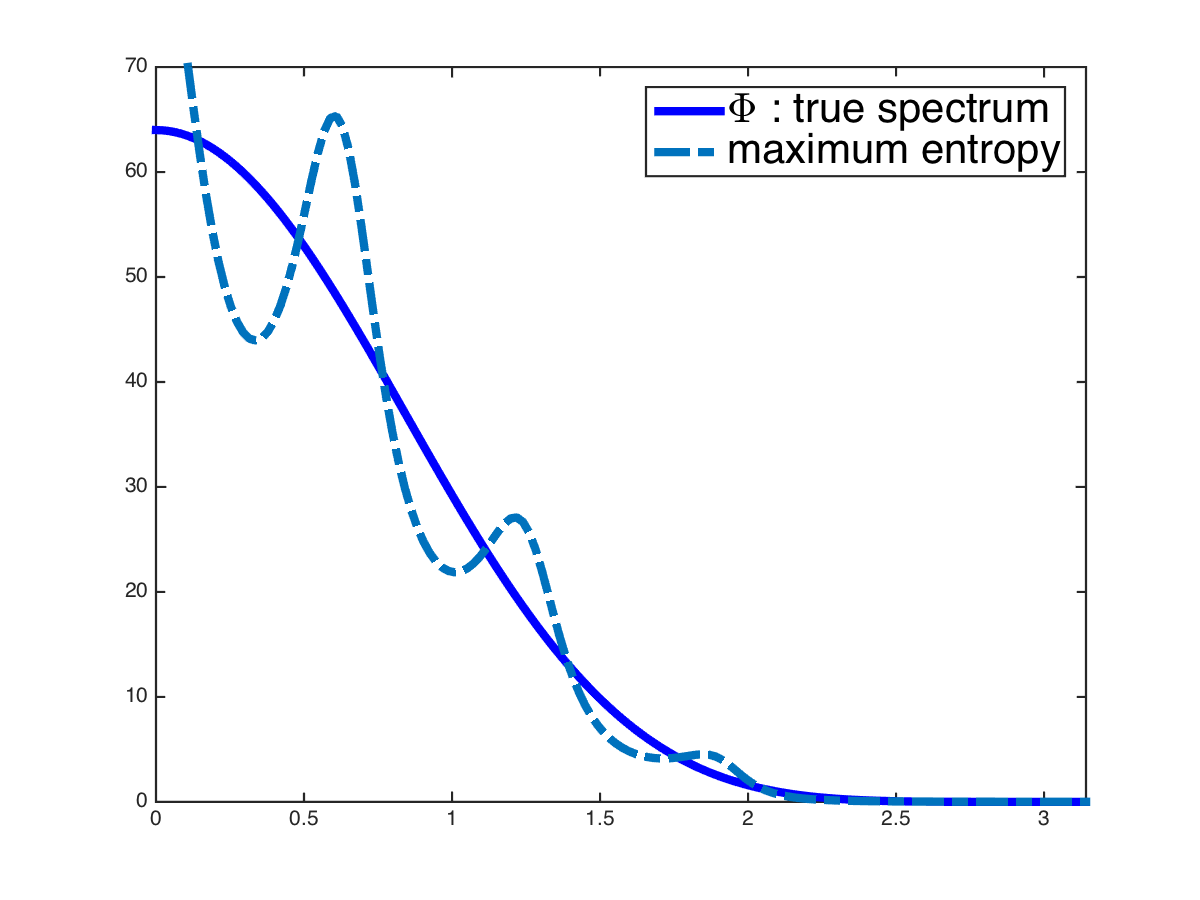}\\
  {\small Fig.\ 4. True spectrum (solid line) vs.\ ME spectrum (dashed)}
\end{center}

 It is interesting to observe the oscillatory character of the all-pole power spectral density, which, of course, coincides with the spectrum obtained following the classical maximum entropy (ME) method \cite{haykin}. In contrast, the use of a prior which corresponds to a low pass filter aliviates the oscillations (Fig.\ 2).

\section{Dual problem and the form of the minimizer}\label{sec:dual}

Suppose that $\Sigma$ belongs to the range of the operator $\Gamma$, defined by \eqref{Gamma}. Then, Problem~1 amounts to minimizing 
\begin{subequations}\label{primal}
\begin{equation}
\label{Dagain}
\mD(\Phi\|\Psi)=\frac12\int \tr\left(\Phi\Psi^{-1}-\log\Phi+\log\Psi-I\right) 
\end{equation}
over all spectral densities $\Phi$ satisfying the moment condition
\begin{equation}
\label{Sigmaagain}
\Sigma=\int G\Phi G^*.
\end{equation}
\end{subequations}
Proceeding  along the lines of \cite{GL1}, it was shown in \cite{FMP} that the dual of \eqref{primal} is the problem to minimize
\begin{subequations}\label{dualproblem}
\begin{equation}
\label{eq:dual}
\mJ(\Lambda) = \tr\left\{\Lambda\Sigma - \int \log Q \right\}
\end{equation}
over all real, symmetric $n\times n$  matrices $\Lambda$ in the range of $\Gamma$ such that 
\begin{equation}
\label{eq:Q}
Q(z):=\Psi(z)^{-1}+G(z)^*\Lambda G(z)
\end{equation}
\end{subequations}
is positive on the unit circle.
For the convenience of the reader, we also review some steps in the proof in our present notation.
We denote the class of feasible $\Lambda$ by $\mathcal{L}_+$, i.e., 
\begin{displaymath}
\mathcal{L}_+ =\{\Lambda\in\text{range}(\Gamma)\mid \Lambda'=\Lambda;\,Q(e^{i\theta})>0,\, \forall \theta\}.
\end{displaymath}
We note in passing that the rationality of  $G$ is not needed at this point; in fact,
 an interesting example with $G(e^{i\theta})=[1,\; e^{i\theta},\;e^{i\sqrt{2}\theta}]'$ is motivated in the context of sensor array processing in \cite{geo6}. 

The Lagrangian for the problem above becomes 
\begin{equation}\nonumber
\label{eq:Lagrangian}
\begin{split}
L(\Phi,\Lambda)&=\mD(\Phi\|\Psi)+\tr\left\{\Lambda\left(\int G\Phi G^* -\Sigma\right)\right\},\\
			&= -\tr(\Lambda\Sigma)+ \int \tr\left\{ \Phi (\Psi^{-1}+G^*\Lambda G)\right.\\
			& \phantom{xxxxxxxxxxxxx}\left.-\log\Phi +\log\Psi -I\right\},
\end{split}
\end{equation}
where $\Lambda$ is a symmetric $n\times n$ matrix of Lagrange multipliers. Since $\tr\left\{\Lambda\left(\int G\Phi G^* -\Sigma\right)\right\}$ is simply the inner product of $\Lambda$ with elements in the range of $\Gamma$, we can restrict $\Lambda$ to the same space and therefore assume that
\[
\Lambda\in\text{range}(\Gamma).
\]
The function $\Phi\mapsto L(\Phi,\Lambda)$ is strictly convex for each $\Lambda$ such that $Q$, defined by \eqref{eq:Q},  is positive semidefinite on the unit circle. If $Q$ fails to be positive semidefinite, $L(\Phi,\Lambda)$ can be made arbitrarily small for some $\Phi$, and hence such a $\Lambda$ is not a candidate in the dual problem. Hence we may restrict $\Lambda$ to the class $\mathcal{L}_+$. Setting the directional derivative
\begin{displaymath}
\delta L(\Phi,\Lambda;\delta\Phi)=\int\tr\left\{ (\Psi^{-1}+G^*\Lambda G -\Phi^{-1})\delta\Phi\right\}
\end{displaymath}
equal to zero, we obtain 
\begin{equation}
\label{eq:optimalPhi}
\Phi = (\Psi^{-1}+G^*\Lambda G)^{-1},
\end{equation}
which inserted into the Lagrangian yields the dual functional
\begin{displaymath}
\begin{split}
\varphi(\Lambda) &= -\tr(\Lambda\Sigma)\\
&\phantom{xxx} + \int \tr\left\{\log (\Psi^{-1}+G^*\Lambda G) +\log\Psi\right\}\\
			   &= -\mJ(\Lambda) + \int \tr\log\Psi.
\end{split}
\end{displaymath}
Since this dual functional should be maximized, the dual problem is equivalent to minimizing $\mJ$ over all $\Lambda\in \mathcal{L}_+$. It was shown in \cite{FMP} that this problem has a unique solution. This problem differs from the one in \cite{GL1} in that the prior $\Psi$ in \cite{GL1} does not occur in $Q$ but instead multiples $\log Q$.  Unlike the situation in \cite{GL1}, $\tr(\Lambda\Sigma)$ might negative in the present setting which complicates the analysis somewhat. 

We shall need the following lemma in Section~\ref{sec:fast}.

\begin{lemma}\em \label{lem:lowerbound}
If $\Sigma$ belongs to the range of $\Gamma$, then the functional $\mJ$ is bounded from below.
\end{lemma}

\begin{proof}
The condition that $\Sigma$ belongs to the range of $\Gamma$ ensures the existence of a spectral density $\Phi_0$ satisfying \eqref{eq:moment}. Then, in view of the construction above, $\varphi(\Lambda)\leq L(\Phi_0,\Lambda)=\mD(\Phi_0\| \Psi)$ or equivalently
\begin{displaymath}
\mJ(\Lambda)\geq \int \tr\log\Psi -\mD(\Phi_0\| \Psi),
\end{displaymath}
which establishes the required bound. 
\end{proof}

%
%

\section{A fast algorithm for the dual problem}\label{sec:fast}

One of the difficulties dealing with the dual problem in Section~\ref{sec:dual} is the redundancy introduced by the integral operator $\Gamma$, which has the consequence that only the part of $\Lambda$ belonging to the range of $\Gamma$ affects the value of $\mJ(\Lambda)$. To remove this redundancy we reformulate the problem by defining  $\mathbb{R}^{n\times m}$ matrix-valued variable
\begin{equation}
\label{eq:X}
X=MB,
\end{equation}
where $M$ is the unique solution of the Lyapunov equation
\begin{equation}
\label{eq:M}
M=A'MA + \Lambda 
\end{equation}
and $\Lambda\in\range(\Gamma)$ is the is the matrix-valued variable in the dual problem in Section~\ref{sec:dual}. By Proposition~\ref{prop:bijection2}, there is a one-one correspondence between $\Lambda$ and $X$.  In view of \eqref{eq:Lambda2X}, 
\begin{displaymath}
G(z)^*\Lambda G(z)=G_0(z)^*X + X'G_0(z),
\end{displaymath}
where  $G_0$ is given by \eqref{eq:Gtilde}. Therefore \eqref{eq:Q} takes the form
\begin{equation}
\label{eq:Q(X)2}
Q(z)=\Psi(z)^{-1}+G_0(z)^*X + X'G_0(z). 
\end{equation}
Moreover, in view of \eqref{eq:M} and \eqref{eq:Sigma}, 
\begin{displaymath}
\begin{split}
\tr(\Lambda\Sigma) &= \tr(M\Sigma) - \tr(MA\Sigma A')\\
                                &=\tr(MBH) +\tr(B'MH')\\
                                &= \tr (HX+X'H').
\end{split}
\end{displaymath}
Consequently the dual functional can be expressed in terms of $X$ to obtain the functional $\mathbf{J}(X):\, \mathcal{X}_+\to\mathbb{R}$ defined by  \eqref{J(X)}, where $\mathcal{X}_+$ is a convex set. 

\

%


To prove that the functional \eqref{J(X)} has a unique minimizer in $\mathcal{X}_+$ we could now appeal to the proof in \cite{FMP} that the dual problem in Section~\ref{sec:dual} has a unique solution. However, since now the redundancy in the dual problem has been removed, we can offer a more straight-forward alternative proof. We denote by $\bar{\mathcal{X}}_+$ the closure of ${\mathcal{X}}_+$.

\

\begin{lemma}\em \label{lem:bounded}
Suppose that $\Sigma$ belongs to the range of $\Gamma$.
Then any nonempty sublevel set 
\begin{equation}
\label{eq:sublevelset}
\{ X\in\bar{\mathcal{X}}_+\mid \mathbf{J}(X)\leq r\}
\end{equation}
is bounded.
\end{lemma} 

\

\begin{proof}
Let $X\in\bar{\mathcal{X}}_+$  be arbitrary, and define $\lambda :=\|X\|$. We want to show that $X$ cannot remain in the level set \eqref{eq:sublevelset} as  $\lambda\to \infty$. To this end, it is no restriction to assume that $\lambda\geq\lambda_0>0$. 
Next set $\tilde{X}:=\lambda^{-1}X$ and $\tilde{Q}_\lambda:=(\lambda\Psi)^{-1}+G_0^*\tilde{X} + \tilde{X}'G_0$.  Then 
\begin{displaymath}
\mathbf{J}(X)=\gamma\lambda - \log\lambda -\tr\int\log \tilde{Q}_\lambda,
\end{displaymath}
where $\gamma:=\tr(H\tilde{X}+\tilde{X}'H')$,  and where $\tilde{Q}_\lambda$ depends on $\lambda$ but is bounded for $\lambda\geq\lambda_0$. First suppose $\gamma >0$. Then comparing linear and logarithmic growth, $\mathbf{J}(X)\to\infty$ as $\lambda\to\infty$, which contradicts $\mathbf{J}(X)\leq r$. Next, suppose that $\gamma\leq 0$. Then $\mathbf{J}(X)\to-\infty$ as $\lambda\to\infty$, which contradicts Lemma~\ref{lem:lowerbound}, since $\mathbf{J}(X)=\mathbf{J}(L(\Lambda)B)=\mJ(\Lambda)$, where $L(\Lambda)$ is the unique solution of the Lyapunov equation \eqref{eq:M}. Hence the sublevel set \eqref{eq:sublevelset} is bounded as claimed. 
\end{proof}

\

\begin{lemma}\em \label{lem:uniquemin}
The functional $\mathbf{J}:\,\bar{\mathcal{X}}_+\to\mathbb{R}\cup\{\infty\}$ has a unique minimizer $\hat{X}$ in ${\mathcal{X}}_+$.
\end{lemma}

\

\begin{proof}
We first prove that $\mathbf{J}$, which is continuous on $\mathcal{X}_+$, can be extended as a lower semicontinuous function $\mathbf{J}:\,\bar{\mathcal{X}}_+\to\mathbb{R}\cup\{\infty\}$. To this end, let $(X_k)$ be a sequence converging to $X$ in $L_\infty$ norm, and let $(Q_k)$ and $Q$ be the corresponding functions \eqref{eq:Q(X)}, which are continuous on the compact interval $[-\pi,\pi]$, and hence uniformly continuous. Consequently there is a bound $\kappa$ such that, for $\theta\in[-\pi,\pi]$, $Q(e^{i\theta})\leq \kappa$ and $Q_k(e^{i\theta})\leq \kappa$ for all $k$, and hence, by Fatou's lemma,
\begin{displaymath}
-\int\log \left( \frac{Q}{\kappa} \right)  \leq \liminf_{k \to \infty} -\int \log \left( \frac{Q_k}{\kappa} \right) 
\end{displaymath}
since $Q_k \rightarrow Q$ pointwise. Consequently, $\mathbf{J}(X)\leq\liminf_{k\to\infty}\mathbf{J}(X))$, which shows that  that $\mathbf{J}$, extended to the boundary $\bar{\mathcal{X}}_+$, is lower semicontinuous. Therefore it follows from Lemma~\ref{lem:bounded} that the sublevel set \eqref{eq:sublevelset} is closed and hence bounded. Consequently, by Weierstrass' Theorem, $\mathbf{J}$ has a minimum $\hat{X}$ in $\mathcal{X}$, which must be unique by convexity.  
%
%

 It remains to prove that $\hat{X}$ is not the boundary $\partial\mathcal{X}$. To this end, following \cite{BGuL,SIGEST}, consider the directional derivative
\begin{displaymath}
\begin{split}
\delta\mathbf{J}(X,\delta X)&= \tr\left\{(H\delta X+\delta X'H') \phantom{\int}\right.\\
&\left.\phantom{xxxxx}- \int Q^{-1}\left(G_0^*\delta X +\delta X'G_0\right)\right\}\\
&=\tr\left\{(H\delta X+\delta X'H') - \int Q^{-1}\delta Q\right\}
\end{split}
\end{displaymath}
Now, for any $X\in\mathcal{X}_+$ and $\bar{X}\in \partial\mathcal{X}$, take $\delta X=X-\bar{X}$ and   $X_\lambda = \bar{X} +\lambda\delta X$  and, correspondently,  form $\delta Q=Q(z)-\bar{Q}(z)$ and   $Q_\lambda(z) = \bar{Q}(z) +\lambda\delta Q_\lambda(z)$, where $\det\bar{Q}(e^{i\theta_0})$ for some $\theta_0\in [-\pi,\pi]$. Then
\begin{displaymath}
\delta\mathbf{J}(X_\lambda,-\delta X)= -\tr(H\delta X+\delta X'H') + \int f_\lambda ,
\end{displaymath}
where $f_\lambda$ is the scalar function
\begin{displaymath}
f_\lambda(e^{i\theta})=\tr\{ Q_\lambda(e^{i\theta})^{-1}\delta Q(e^{i\theta})\}.
\end{displaymath}
Taking the derivative with respect to $\lambda$ we have 
\begin{displaymath}
\frac{d}{d\lambda}f_\lambda(e^{i\theta})=\tr\{\delta Q(e^{i\theta})^*Q_\lambda(e^{i\theta})^{-2}\delta Q(e^{i\theta})\}\geq 0,
\end{displaymath}
and consequently $f_\lambda(e^{i\theta})$ is a monotonically nondecreasing function of $\lambda$ for all $\theta\in [-\pi,\pi]$. Therefore, as $\lambda\to 0$,  $f_\lambda$ tends pointwise to 
\begin{displaymath}
\begin{split}
f_0&=\tr\{\bar{Q}^{-1}(Q-\bar{Q})\}=\tr\{ \bar{Q}^{-1}Q -I\}\\
&=\tr\{ \bar{Q}^{-1}Q \} -n.
\end{split}
\end{displaymath}
 If $\int f_\lambda$ would tend to a finite value as $\lambda\to 0$, $(f_\lambda)$ would be a Cauchy sequence in $L^1(-\pi,\pi)$ and hence have a limit in $L^1(-\pi,\pi)$ equal almost everywhere to $f_0$. However,  since there is a $\delta >0$ such that $Q(e^{i\theta})>\delta$, 
 \begin{displaymath}
\int f_0 \geq  \delta\int \tr\left(\bar{Q}^{-1}\right) -n
\end{displaymath}
 which is infinite by Proposition~\ref{Qprop}.
Consequently
\begin{displaymath}
\delta\mathbf{J}(X_\lambda,\bar{X}-X)\to\infty\quad\text{as $\lambda\to 0$},
\end{displaymath}
so there could be no minimum in $\bar{X}$. This concludes the proof. 
\end{proof}

Since the unique minimizer $\hat{X}$ belongs to the interior $\mathcal{X}_+$, the gradient \eqref{eq:gradient} is zero there. This proves \eqref{eq:Xhat}. Then, by \eqref{eq:optimalPhi}, the optimal solution of Problem~\ref{problem1} is given by \eqref{eq:optimalPhiy}. This concludes the proof of Theorem~\ref{thm1}. 

\section{closed-form solution for a special case of prior}\label{sec:closedform}

We now consider the special case where the prior power spectral density is of the particular form
\[
\Psi=(G(z)^*\Lambda_0 G(z))^{-1}.
\]
Then the matrix function $Q$ defined by \eqref{eq:Q} is given by
\begin{displaymath}
Q(z) = G(z)^*(\Lambda_0+\Lambda)G(z),
\end{displaymath}
which must be positive on the unit circle and hence, by Lemma \ref{lemma1} there exists a constant matrix $C$ such that
\begin{equation}
\label{eq:Qnew}
Q(z)=G(z)^*CC'G(z).
\end{equation}

We first change the dual functional \eqref{eq:dual} by adding the constant
$\tr(\Lambda_0\Sigma)$, and compute
\begin{equation}
\label{eq:eq2}
\begin{split}
\tr((\Lambda+\Lambda_0)\Sigma) &= \tr\int (\Lambda+\Lambda_0)G\Phi G^*\\
&= \tr\int Q\Phi\\
&=\tr\int C'G\Phi G^*C\\
&=\tr C'\Sigma C,
\end{split}
\end{equation}
where $\Phi$ satisfies \eqref{eq:moment}. In view of \eqref{eq:Qnew}, the modified functional becomes
\begin{equation}
\label{eq:modJ}
\begin{split}
\tilde{\bJ}(C)&:=\mJ(\Lambda)+\tr \Lambda_0\Sigma\\
&=\tr \left(C'\Sigma C- \int \log G^*CC'G\right)
\end{split}
\end{equation}
which is now a function of $C$. 
Recall the following result from Wiener-Masani-Helson-Lowdenslager.

\

\begin{prop}\em If $F(z)$ is a square outer matrix-valued function, then
\begin{displaymath}
\int \log \det FF^* = \log \det F(0)F(0)^*.
\end{displaymath}
\end{prop}
\begin{proof} The result follows by Jensen's formula after noting that $f=\det F$ is outer (\cite[p. 184]{Ahlfors}).
\end{proof}

\

We now consider once again the functional $\tilde{\bJ}(C)$ and determine stationarity conditions that provide a form of the optimal $C$. First,
\begin{equation}
\nonumber 
\begin{split}
\tilde{\bJ}(C) & =  \tr \left(C'\Sigma C\right)- \log\det(B'CC'B)\\
&=\tr \big(C'\Sigma C- \log(B'CC'B)\big).
\end{split}
\end{equation}
The gradient with respect to $C$ is
\[
\frac{\partial\tilde{\bJ}}{\partial C}=2C'\Sigma - 2(B'C)^{-1}B',
\]
and hence the stationary point is given by $C'\Sigma =(B'C)^{-1}B'$.  This yield the equation
\begin{equation}
\label{eq:verify}
B'CC' = B'\Sigma^{-1}
\end{equation}
for the optimal $C$, and we readily see that
\[
C=\Sigma^{-1} B(B'\Sigma^{-1}B)^{-1/2}
\]
satisfies \eqref{eq:verify}. Thus, the optimal $Q$ is
\[
\hat Q(z)=G(z)^*\Sigma^{-1}B(B'\Sigma^{-1} B)^{-1}B'\Sigma^{-1}G(z),
\]
and therefore
\[
\hat\Phi(z)=(G(z)^*\Sigma^{-1}B(B'\Sigma^{-1} B)^{-1}B'\Sigma^{-1}G(z))^{-1}.
\]
This concludes the proof of Theorem~\ref{thm2}. 

%
%

\section{Conclusions}\label{sec:conclusions}

The topic of the paper is to construct power spectral densities that are consistent with specified moments
and are closest to a prior in a suitable sense. The spirit of the work is similar to a long line of contributions going
back to \cite{burg}, including a series of papers \cite{BGuL}--\cite{KLR} where the emphasis was in identifying and parametrizing power spectra of minimal complexity (i.e., dimensionality of modeling filters). A key tool in these earlier works was a choice of entropy functional that allowed parametrizing solutions via selection of a suitable prior power spectrum. The moment constraints were cast in the form of the state covariance of an input-to-state filter.

In departure from this early work, Ferrante, Masiero and Pavon \cite{FMP} proposed to use the KL divergence between Gaussian probability laws -- a formulation which is quite natural from a probabilistic standpoint. The KL divergence between Gaussian probability laws coincides with the Itakura Saito distance between their respective power spectral densities, and thus, the problem turns out to be equivalent to one studied by Enqvist and Karlsson \cite{EK} in the context of scalar processes.
The purpose of the current work is to present a simplified alternative optimization procedure which is based on a detailed analysis of the geometry of input-to-state filters and related moment problems. Indeed, the power spectral densities are now parametrized more conveniently by a non-redundant coefficient matrix ($X$ in Theorem \ref{thm1}) containing minimal number of parameters that are necessary.
Sections \ref{sec:geometry} and \ref{mainresults}, as well as the proofs later in the paper contain the main contributions.

\section{Appendix}

\subsection{Behavior of\/ $\mathbf{J}$ on the boundary}

\begin{lemma}\em \label{Mlemma}
Let $\theta\mapsto M(e^{i\theta})$ be a matrix-valued function with Lipschitz-continuous components, and suppose that $M(e^{i\theta})$ is positive semidefinite for all $\theta$ and identically zero for $\theta=\theta_0$. Then
\begin{displaymath}
\int_{-\pi}^\pi \tr\{M^{-1}(e^{i\theta})\}\frac{d\theta}{2\pi} =\infty,
\end{displaymath}
where $M^{-1}$ is defined to have infinite value on any subset  of $[-\pi,\pi]$ where it is identically zero.
\end{lemma}

\

\begin{proof}
Without loss of generality we can assume that $M(e^{i\theta})=0$ in  an isolated point $\theta_0$. By assumption, we can choose a common Lipschitz constant $K$ and an $\varepsilon >0$ such that the components $m_{k\ell}(e^{i\theta})$ of $M$ have the bounds
\begin{displaymath}
\left| m_{k\ell}(e^{i\theta}) \right| \leq K |\theta -\theta_0|  
\end{displaymath}
for $ |\theta -\theta_0| <\varepsilon$.
If $N(e^{i\theta}):=M^{-1}(e^{i\theta})$, its components satisfy
\begin{displaymath}
\sum_\ell m_{k\ell}(e^{i\theta})n_{\ell k}(e^{i\theta})=1 \quad \text{for all $\theta$ and $k$},
\end{displaymath}
which then implies that 
\begin{displaymath}
\begin{split}
&\left|\sum_\ell n_{\ell k}(e^{i\theta})\right|K|\theta -\theta_0|\geq 1 \quad \\
&\phantom{xxxxx}\text{for all $\theta\in (\theta_0-\varepsilon,\theta_0+\varepsilon)$ and $k$}.
\end{split}
\end{displaymath}
Consequently, since $N(e^{i\theta})\geq 0$, there must be a $k$ and an $L>0$ such that 
\begin{displaymath}
n_{kk}(e^{i\theta})\geq \frac{1}{L|\theta -\theta_0|}, \quad\text{for all $\theta\in (\theta_0-\varepsilon,\theta_0+\varepsilon)$},
\end{displaymath}
and therefore
\begin{displaymath}
\int_{-\pi}^\pi \tr\{M^{-1}(e^{i\theta})\}\frac{d\theta}{2\pi}\geq \frac{1}{L}\int_{\theta_0-\varepsilon}^{\theta_0+\varepsilon}\frac{1}{|\theta -\theta_0|}\frac{d\theta}{2\pi} =\infty,
\end{displaymath}
as claimed. 
\end{proof}

\

\begin{prop}\em\label{Qprop}
Let $\theta\mapsto Q(e^{i\theta})$ be a matrix-valued function with Lipschitz-continuous components, and suppose that $Q(e^{i\theta})$ is positive semidefinite for all $\theta$ and singular for $\theta=\theta_0$. Then
\begin{displaymath}
\int_{-\pi}^\pi \tr\{Q^{-1}(e^{i\theta})\}\frac{d\theta}{2\pi} =\infty.
\end{displaymath}
\end{prop}

\

\begin{proof}
After applying a constant unitary transformation we can write $Q$ on the form
\begin{displaymath}
Q=\begin{bmatrix} Q_1&Q_2\\Q_2^*&Q_3\end{bmatrix},
\end{displaymath}
where $Q_1(e^{i\theta_0})=Q_2(e^{i\theta_0})=0$ and $Q_3(e^{i\theta_0})>0$. Then
\begin{displaymath}
Q^{-1}=\begin{bmatrix} \left[Q_1-Q_2Q_3^{-1}Q_2^*\right]^{-1}&*\\ *&*\end{bmatrix},
\end{displaymath}
where the Schur complement \[M:=Q_1-Q_2Q_3^{-1}Q_2^*\] is positive semidefinite and has Lipschitz-continuous components. Then the statement of the proposition  follows from Lemma~\ref{Mlemma}. 
\end{proof}

\subsection{Co-invariant subspaces}
Let $\cH_2^m$ represent {\em row} vector-valued functions
in the Hardy space of square integrable functions on the circle which have an analytic continuation in the interior of the unit disc -- a standard notation $\cH_2$ or $\cH_2(\mD)$. The forward shift $S$ amounts to multiplication by
$\la$. The backward shift is precisely its adjoint,
\[
S^*:\cH_2^{\ell}\to \cH_2^{\ell}:x(\la)\mapsto \bPi_{\cH_2}\la^{-1} x(\la).
\]
Subspaces which are invariant under $S^*$ are those that are orthogonal to invariant subspaces
of the forward shift $S$, i.e., of the form
\[
\cK:= \cH_2^{1\times m}\ominus \cH_2^{1\times m} V(\la)
\]
with $V(\la)$ an inner (matrix-valued) function, and they are often referred to simply as
``co-invariant subspaces''. 
The orthogonal projection onto $\cK$ is
\[
\begin{split}
{\bf \Pi}_\cK &: \cH_2^{1\times m} \to \cK\\ & x(\la) \mapsto
\left( {\bf \Pi}_{(\cH_2^{1\times m})^\perp} x(\la)V(\la)^*
\right) V(\la).
\end{split}
\]
To see this, note that since $V(\la)$ is inner,
${\bf \Pi}_\cK$ defined above
is idempotant and Hermitian---hence a projection.
It is easy to verify that its kernel is precisely $\cH_2^{1\times m}V(\la)$
and therefore ${\bf \Pi}_\cK$ is the orthogonal projection onto $\cK$ as claimed.
\vspace*{.1in}

Let $A\in\mC^{n\times n}$ with eigenvalues in $\mD$, $B\in\mC^{n\times m}$ with $(A,B)$ controllable. Without loss in generality we can always normalize $(A,B)$ so that the corresponding controllability Grammian is the identity $I$; when this is true
\[
AA^*+BB^*=I
\]
and $[A,\,B]$ can be competed to a unitary matrix
\[
\left[\begin{matrix}A&B\\C&D\end{matrix}\right].
\]
It follows that
\[
V(\la)=D+zC(I-zA)^{-1}B
\]
is an inner matrix-valued function, i.e., it is analytic in $\mD$ and
$VV^*=V^*V=I$, where $V^*:=V(\la):=V^*(\la^{-1})$. Now, consider
\[
G(\la):=(I-\la A)^{-1}B
\]
and the co-invariant subspace $\cK$ as noted above. The following statement is known (see \cite[Proposition 4]{Georgiou2002}).

\

\begin{prop}\em
The rows of $G(\la)$ form a basis for $\cK$.

\

\end{prop}
\vspace*{.1in}

\begin{proof} The proof is again from \cite[Proposition 4]{Georgiou2002}.
We first claim that any element in $\cK$ is of the form
\begin{equation}
v(\la I - A^*)^{-1}C^*V(\la) \label{formofelements}
\end{equation}
where $v\in\mC^{1\times n}$. To see this note that
\begin{eqnarray*}
{\bf \Pi}_\cK&:& x_0+x_1\la+\ldots \mapsto\\
& &\hspace*{-.3cm} \big[ {\bf \Pi}_{(\cH_2^{1\times m})^\perp}
(x_0+x_1\la+\ldots)\\
&& \times(D^*+\la^{-1}B^*C^*+\ldots)\big]V(\la)\\
& &\hspace*{-.3cm}= v(\la^{-1}C^*+\la^{-2}A^*C^*+\ldots)V(\la)
\end{eqnarray*}
where $v=x_0B^*+x_1B^*A^*+\ldots$.
Next, it can be shown \cite[Eq.\ (36)]{Georgiou2002} that
\begin{equation}
G(\la)=(\la I - A^*)^{-1}C^* V(\la). \label{Gfraction}
\end{equation}
In view of (\ref{formofelements}), the rows of $G(\la)$ span $\cK$.
Finally, if $vG(\la)=0$ for some $v\in\mC^{1\times n}$,
then necessarily $v=0$ because $(A,B)$ is controllable.
Hence the rows of $G(\la)$ are linearly independent and form a basis
for $\cK$ as claimed.
\end{proof}

\

\begin{lemma}\em \label{lemma1} Let $\Lambda$ be a Hermitian $n\times n$-matrix such that
\[Q(\la):=G(\la)^*\Lambda G(\la)>0
\]
for $\la=e^{i\theta}$, and $\theta\in[0,2\pi)$.
There exists $\Lambda_o=C^*_oC_o$ with $C_o\in\mC^{m\times n}$ such that
\[
G(\la)^*\Lambda G(\la)=G(\la)^*\Lambda_o G(\la)
\]
and $C_oG(\la)$ is outer (i.e., minimum phase, that is, stable and stably invertible).
\end{lemma}

\

\begin{proof}
Since $Q(\la)$ is Hermitian and positive definite on the unit circle of the complex plane, it can be factored as
\[
Q(\la)=a(\la)^*a(\la)
\]
with $a(\la)$ outer. But $V(\la)G(\la)^*\Lambda G(\la)=V(\la)a(\la)^*a(\la)$ has all its elements in $\cH_2$, since already $V(\la)G(\la)^*$ does.  Since $a(\la)$ is outer, $V(\la)a(\la)^*$ is in $\cH_2$ as well. Now, note that $G(\la)V(\la)^*$ is orthogonal to $\cH_2$. Therefore, the zeroth term of $G(\la)V(\la)^*$ vanishes. It follows that $V(\la)a(\la)^*a(\la)$ which has only positive power of $\la$ has no $0$th term either. Therefore, $V(\la)a(\la)^*$ has only positive powers of $\la$ and no $0$th term. So, finally, we conclude that all elements of $a(\la)V(\la)^*$ are orthogonal to $\cH_2$ and, therefore, the rows of $a(\la)$ are in $\cK$. Thus, there exists a $C\in\mC^{m\times n}$ such that
\[
a(\la)=CG(\la).
\]
This completes the proof.
\end{proof}

\end{document}